\documentclass[12pt]{article}

\usepackage{url} % not crucial - just used below for the URL 

\newcommand{\blind}{0}

\usepackage{multirow}
\usepackage{fullpage,setspace}
\usepackage{amssymb, amsthm, amsmath}
\usepackage{graphicx}
\usepackage[authoryear]{natbib}
\usepackage{bm}
\usepackage{lineno}
\usepackage{times}
\usepackage[ruled,vlined]{algorithm2e}
\usepackage{float,color,comment}
\usepackage[usenames,dvipsnames]{xcolor}
\usepackage{soul}
\usepackage{hyperref}

\newcommand{\btheta}{ \mbox{\boldmath $\theta$}}

\newcommand{\bbeta}{ \mbox{\boldmath $\beta$}}
\newcommand{\bphi}{\mbox{\boldmath $\phi$}}

\newcommand{\Tr}{^{\rm T}}
\newcommand{\bA}{ \mbox{\bf A}}

\newcommand{\bX}{ \mbox{\bf X}}
\newcommand{\bB}{ \mbox{\bf B}}

\newcommand{\bY}{ \mbox{\bf Y}}

\newcommand{\bs}{ \mbox{\bf s}}

\newcommand{\bQ}{ \mbox{\bf Q}}
\newcommand{\bR}{ \mbox{\bf R}}

\newcommand{\bK}{\mbox{\bf K}}

\newcommand{\iid}{\stackrel{iid}{\sim}}

\newcommand{\calN}{{\cal N}}

\newcommand{\Matern}{ \mbox{Mat$\acute{\mbox{e}}$rn}}

\newcommand{\beq}{ \begin{equation}}
\newcommand{\eeq}{ \end{equation}}
\newcommand{\beqn}{ \begin{eqnarray}}
\newcommand{\eeqn}{ \end{eqnarray}}

\newtheorem{theorem}{Theorem}

\DeclareMathOperator{\Trace}{\mathrm{Tr}}

\usepackage{caption}
\renewcommand{\arraystretch}{1.5}

\DeclareMathOperator{\EX}{\mathbb{E}}

\newcommand{\mbsize}{n_{\mathcal{B}}}

\begin{document}  %\linenumbers

\def\spacingset#1{\renewcommand{\baselinestretch}%
{#1}\small\normalsize} \spacingset{1}

%%%%%%%%%%%%%%%%%%%%%%%%%%%%%%%%%%%%%%%%%%%%%%%%%%%%%%%%%%%%%%%%%%%%%%%%%%%%%%

\if0\blind
{
  \title{\bf Stochastic Gradient MCMC for Massive Geostatistical Data}
  \author{Mohamed A. Abba
  %\thanks{
   % This research was partially supported by National Science Foundation grants DMS2152887 and CMMT2022254, and by grants from the Southeast National Synthesis Wildfire and the United States Geological Survey’s National Climate Adaptation Science Center (G21AC10045).}
    \hspace{.2cm}\\
    Department of Statistics, North Carolina State University\\
    Brian J. Reich\\
    Department of Statistics, North Carolina State University \\
    Reetam Majumder \\
    Southeast Climate Adapatation Science Center, North Carolina State University \\
    Brandon Feng\\
    Department of Statistics, North Carolina State University\\
    }
  \maketitle
} \fi

\if1\blind
{
  \bigskip
  \bigskip
  \bigskip
  \begin{center}
    {\LARGE\bf Title}
\end{center}
  \medskip
} \fi

\bigskip
\begin{abstract}

Gaussian processes (GPs) are commonly used for prediction and inference for spatial data analyses. However, since estimation and prediction tasks have cubic time and quadratic memory complexity in number of locations, GPs are difficult to scale to large spatial datasets. The Vecchia approximation induces sparsity in the dependence structure and is one of several methods proposed to scale GP inference. Our work adds to the substantial research in this area by developing a stochastic gradient Markov chain Monte Carlo (SGMCMC) framework for efficient computation in GPs. At each step, the algorithm subsamples a minibatch of locations and subsequently updates process parameters through a Vecchia-approximated GP likelihood. Since the Vecchia-approximated GP has a time complexity that is linear in the number of locations, this results in scalable estimation in GPs. Through simulation studies, we demonstrate that SGMCMC is competitive with state-of-the-art scalable GP algorithms in terms of computational time and parameter estimation. An application of our method is also provided using the Argo dataset of ocean temperature measurements.
\end{abstract}
\noindent%
\vfill

%\newpage
\spacingset{1.75} % DON'T change the spacing!
%\newpage

\section{Introduction}\label{s:intro}

Gaussian process (GP) modeling is a powerful statistical and machine learning tool used to tackle a variety of  tasks including regression, classification, and optimization. Within spatial statistics, in particular, GPs have become the primary tool for inference \citep{GELFAND201686}. 
In spatial regression and classification problems, the response variable is assumed to have a spatially-correlated structure. GPs model this spatial dependence by specifying a form for the correlation between any two points in the spatial domain. In this paper we focus on the regression setting under the{\Matern} correlation with large amounts of data. Formally, GPs place a prior on the spatial process using a parameterized correlation function, which allows us to estimate \textit{a posteriori} the parameters given the observed data.

One of the main advantages of GPs is their ability to provide predictions at unobserved locations along with uncertainty quantification. Spatial interpolation, commonly known as Kriging \citep{woodard2000interpolation}, provides optimal predictions at unobserved sites based on the correlation between a given location and its observed neighbours \citep{cressie1988spatial}. However, handling large datasets with GPs poses computational challenges due to the cubic time complexity and quadratic memory requirements for the evaluation of the joint likelihood. This prohibitive computational cost mainly results from the evaluation of the covariance matrix and computing its inverse. Several methods have been proposed to address this issue and make GPs more scalable for large datasets. In this work, we combine  stochastic gradient (SG) methods along with the Vecchia \citep{vecchia1988estimation} approximation to develop an efficient algorithm for scalable Bayesian inference in massive spatial data settings. In the following section we review some of the main methods used to scale GPs \citep[see][for a full survey]{heaton2019case}, and briefly discuss applications of SG methods in correlated and dependent data settings.

\subsection{Methods to handle large spatial datasets}
The main computational bottleneck in GP regression is evaluating the inverse of the covariance matrix. To overcome this problem, a large body of literature has been proposed over the last decades, including but not limited to low rank approximations, covariance tapering, divide-and-conquer strategies, and Vecchia-type methods. Although these approaches differ significantly, they all tend to result in an amenable structure on the covariance or its inverse.

A low-rank approximation of the GP can be used to overcome the covariance inverse cost, (\textit{e.g.}, \cite{cressie2008fixed, katzfuss2011spatio, kang2009statistical}). Low-rank approximations project the spatial process on a low-dimensional space and use the low-rank representation as a surrogate to approximate the original process. \cite{banerjee2008gaussian} used predictive process methods, where first a certain number of knots are placed in the spatial domain, then used as a conditioning set for the expectation of the original process. Fixed rank kriging \citep{cressie2008fixed} approximates the original process using a small number of basis functions, which results in a precision matrix that can be obtained by inversion of a matrix with a much smaller dimension.

Instead of approximating the original process, one can impose fixed structures on the covariance or precision matrices directly. This method, also known as covariance tapering \citep{furrer2006covariance, kaufman2008covariance}, imposes a compact support on the correlation function, and hence correlation between a site and distant neighbours is shrunk to zero. This induces a sparse structure on the covariance that is leveraged to speed up the computation. Instead of imposing a structure on the covariance, \cite{rue2009approximate} directly impose a sparse structure on the precision matrix using a Gaussian Markov random field approximation to the true process.

Divide-and-conquer approaches have also been proposed to scale GPs inference. \cite{barbian2017spatial} and \cite{guhaniyogi2018meta} propose splitting the spatial domain in subsets, performing the analysis in parallel on each subset, and then combining the results. This strategy distributes the workload into smaller parts. Another option is to divide the spatial region into independent sub-regions and perform the analysis on the whole dataset under this assumption \citep{sang2011covariance}. Unlike the former, the latter uses the whole dataset but reduces the computational cost using independence between subregions.

One of the earliest and most influential methods for scalable GPs is the Vecchia approximation \citep{vecchia1988estimation, stein2004approximating}. In the Vecchia framework, the full likelihood is factorized into a series of conditional distributions. This factorization is then simplified by reducing the conditioning sets to include a small number of neighbours, which in turn results in a sparse precision matrix. \cite{guinness2018_kl} showed that  Vecchia's method is an accurate approximation to the true Gaussian process model in terms of the Kullback-Leibler divergence. This approach is also well suited for parallel computing due to the factorization of the likelihood. Recent works have built upon and extended the Vecchia approximation. \cite{katzfuss2021general} propose a generalization of the Vecchia's framework and show that many existing approaches to Gaussian process approximation can be viewed as a special case of the extended  method. \cite{datta2016a} proposed the nearest-neighbor Gaussian process as an extension of the Vecchia approximation, later, \cite{ finley2019efficient} outline an efficient Markov Chain Monte Carlo (MCMC) algorithm for scalable full Bayesian inference using this method.

In general, all the aforementioned methods reduce the computational cost from cubic to linear in the number of observations. However, in  the Bayesian framework we are mostly interested in posterior sampling through MCMC methods in order to get uncertainty estimates of the model parameters as well predictive credible intervals for certain locations. Typically, MCMC methods require thousands of iterations to accurately approximate the posterior distribution. Hence, even when the cost per iteration is linear, the total time can still be prohibitive. Recent work has therefore also focused on subsampling approaches for spatial data to reduce the computational cost associated with posterior sampling. 
 \citet{saha2023minibatch} have developed an efficient composite sampling scheme for posterior inference. Similarly, \citet{heaton2023minibatch} use minibatches to approximate the complete conditional distribution of conjugate parameters, and provide an approximate Metropolis-Hastings (MH) acceptance probabilities for non-conjugate parameters. While \citet{heaton2023minibatch} use a Vecchia approximation to define the minibatches, neither of these two works use any gradient information when drawing samples from the posterior, and are therefore fundamentally different from the gradient-based approach we will employ in our study.

\subsection{Review of stochastic gradient methods}
When dealing with large datasets, stochastic gradient (SG) methods \citep{Robbins1951sdg} have become the default choice in machine learning \citep{hardt2016train}. To avoid computing a costly gradient based on the full dataset, SG methods only require an unbiased and possibly noisy estimate using a subsample of the data. When the data is independent and identically distributed (\textit{iid}) a proper scaling of the gradient based on a given subsample of the data yields an unbiased gradient estimate. The popularity and success of SG methods in optimization eventually lead to their adoption for scalable Bayesian inference \citep{Nemeth2021sgmcmc}. Scalable SG Markov Chain Monte Carlo (SGMCMC) methods for posterior sampling in the \textit{iid} setting have been proposed \citep{Welling2011BayesianLV, Chen2015OnTC, Ma2015ACR, dubey2016variance, baker2019control}. Convergence of SGMCMC methods has also received considerable attention. Under mild conditions, SGMCMC methods produce approximate samples from the posterior \citep{teh2016consistency, durmus2017nonasymptotic, dalalyan2019user}.

Although SG methods are widely used in the \textit{iid} setting, their possible use in the correlated setting is still new. A naive application of SGMCMC methods in the correlated setting would overlook critical dependencies in the data during subsampling. Moreover, the gradient estimate from the subsamples cannot be guaranteed to be unbiased. To the best of our knowledge, subsampling methods for spatial data that result in unbiased gradient estimates has not been addressed. \cite{chen2020stochastic} studied the performance and theoretical guarantees for SG optimization for GP models. Although, the gradient based on a minibatch of the data leads to biased estimates of the full gradient of the log-likelihood, \cite{chen2020stochastic} established convergence guarantees for recovering recovering noise variance and spatial process variance in the case of the exponential covariance function. In their work, the length scale parameter, which controls the degree of correlation between distinct points is assumed known, and no convergence result is provided. Recent works have considered other types of dependent data. In the case of network data, \cite{li2016scalable} developed an SGMCMC algorithm for the mixed-member stochastic block models. \cite{ma2017stochastic} leveraged the short-term dependencies in hidden Markov models to construct an estimate of the gradient with a controlled bias using non-overlapping subsequences of the data. This approach was extended to linear and non-linear state space models \citep{aicher2019stochastic, aicher2021stochastic}.

SGMCMC methods can be divided in two main groups based on either Hamilton dynamics \citep{chen2014stochastichmc} or Langevin dynamics \citep{Welling2011BayesianLV}. In this work we use the Langevin dynamics (LD) method due to its lower number of hyperparameters, our approach can be extended to the Hamiltonian dynamics with minor modifications. We extend the SGLD method to the case of non-\textit{iid} data using the Vecchia approximation and provide a method that takes account of the local curvature to improve convergence.

In the remainder of this paper, Section \ref{s:model} discusses the{\Matern} Gaussian process model and the Vecchia approximation used to obtain unbiased gradients. Section \ref{s:comp} presents the derived SGMCMC algorithm for Gaussian process learning. We test our proposed method using a simulation study in Section \ref{s:sim}, and present a case study for ocean temperature data in Section \ref{s:argo}; Section \ref{s:discussion} concludes. A modification of our approach into a stochastic gradient Fisher scoring method for GPs is discussed in the Supplementary Material, alongside its performance for maximum likelihood estimation.

\section{$\Matern$ Gaussian Process Model and its Approximations}\label{s:model}

Let $Y_i$ for $i\in\{1,...,n\}$ be the observation at spatial location $\bs_i = (s_{i1},s_{i2})$ and $\bX_i = (X_{i1},...,X_{ip})$ be a corresponding vector of covariates. The data-generation model for Gaussian process regression in the case of Gaussian data is
%Gaussian process regression for $Y_i$ assumes that the conditional mean  $\EX(Y_i|\bX_i,Z(\bs_i))=\bX_i\bbeta + Z(\bs_i)$ is  linear for fixed effects $\bbeta$ and spatial process $Z(\bs_i)$. The data generation model, for Gaussian data is
 \begin{equation}\label{e:Y}
  Y_i =  \bX_i\bbeta + Z(\bs_i) + \varepsilon_i,
  %\mbox{Binary} && \mbox{logit}\{\mbox{Prob}(Y_i=1|\et\bA_i)\} = \bX_i\bbeta + Z(\bs_i)\nonumber
 \end{equation} 
with fixed effects $\bbeta$, spatial process $Z(\bs_i)$ and $\varepsilon_i\iid\mbox{Normal}(0,\tau^2)$ is measurement error  with nugget $\tau^2$. 
The process $Z(\bs)$ is an isotropic spatial Gaussian process with mean $\mbox{E}\{Z(\bs)\}=0$, spatial variance $\mbox{Var}\{Z(\bs)\}=\sigma^2$ and spatial correlation $\mbox{Cor}\{Z(\bs_i),Z(\bs_j)\}={\rm{K}} (d_{ij})$ for distance $d_{ij}=||\bs_i-\bs_j||$.  Specifically, we assume the correlation function is the $\Matern$ \citep{stein1999interpolation} correlation function with range $\rho$ and smoothness $\nu$
\begin{equation}\label{eq: matern_cov}
    {\rm{K}}(d) = \frac{1}{\Gamma(\nu)2^{\nu-1}}\left(\frac{d}{\rho} \right)^{\nu} \mathcal{K}_{\nu}\left(\frac{d}{\rho} \right),
\end{equation}
where $\mathcal{K}_\nu$ is the modified Bessel function of the second kind. Let $\btheta = (\sigma^2, \rho, \nu, \tau^2)$ be the collection of covariance parameters. 

The marginal distribution (over $Z$) of $\bY = \left\{Y(\bs_1),\ldots, Y(\bs_n)\right\}$ is multivariate normal with mean $\EX[\bY] = \bX\bbeta$, for $\bX \in \mathbf{R}^{n\times p}$ covariate matrix with the $\rm{i}^{\rm{th}}$ row $\bX_i$, and covariance matrix $\EX[(\bY - \bX\bbeta)(\bY - \bX\bbeta)\Tr \mid \btheta] = \Sigma(\btheta)$ with
\begin{align}
    \Sigma(\btheta) &= \sigma^2\bK + \tau^2\mathbf{I}_n, \\
    \bK_{i,j} &= {\rm{K}}(d_{ij}). \nonumber
\end{align}
The full log-likelihood then becomes
\begin{equation}\label{eq: gp_full_llk}
    \ell_{\rm{full}}(\bbeta, \btheta) = -\frac{n}{2}\log(2\pi) - \frac{1}{2}\log \det \Sigma(\btheta) - \frac{1}{2} (\bY - \bX\bbeta)\Tr\Sigma(\btheta)^{-1}(\bY-\bX\bbeta).
\end{equation}
Evaluating \eqref{eq: gp_full_llk} involves computing the determinant and inverse of $\Sigma(\btheta)$ which generally requires $O(n^3)$ operations. This cost becomes prohibitive for large spatial datasets. The remainder of this section discusses the computationally-efficient Vecchia approximation.

\subsection{The Vecchia approximation}\label{ss:vecchia}

For any set of spatial locations, the joint distribution of $\bY$ can be written as a product of univariate conditional distributions, which can then be approximated by a Vecchia approximation \citep{vecchia1988estimation,stein2004approximating,datta2016a,katzfuss2021general}:
\begin{equation}\label{e:vecchia}
    f(Y(\bs_1),...,Y(\bs_n)) = \prod_{i=1}^n f(Y(\bs_i)|Y(\bs_1), ...,Y(\bs_{i-1}))
    \approx
    \prod_{i=1}^n f_i(Y(\bs_i)|Y(\bs_{(i)})),
\end{equation}
for $Y(\bs_{(i)}) = \{Y(\bs_j); j\in \calN_i\}$ and conditioning set ${\cal N}_i\subseteq\{1,...,i-1\}$, e.g., the indices of the $m_i\leq m$ locations in ${\cal N}_i$ that are closest to $\bs_i$ according to some ordering of the data. Here, we use the notation that the collection of variables over the conditioning set of $Y_i$ is denoted $Y_{(i)} = \{Y_j; j\in \calN_i\}$.  Of course, not all locations that are dependent with location $i$ need be included in $\calN_i$ because distant observations may be approximately independent after conditioning on more local observations. Conditioning on only $\mathcal{N}_i$ leads to substantial computational savings when $m$ is small, \textit{i.e.}, $m<<n$. The Vecchia approximation is attractive for GPs in particular since the conditional densities are Normally distributed. \citet{Stein2002Screening} proved that a screening effect exists in this context which ensures that the Vecchia approximation is a good approximation, and \cite{Stein2011Screening} provided conditions  for a variety of situations when the screening effect might hold. Also, while we have motivated the Vecchia likelihood as an approximation, it is in fact a valid joint probability density function (PDF) which permits a standard Bayesian analysis and interpretation. 

Let $p(\bbeta,\btheta)$ be the prior distribution on the regression and covariance parameters. Using \eqref{e:vecchia} we can write the posterior as (ignoring a constant that does not depend on the parameters)
\begin{align}\label{eq: vecchia_llk_repsonse}
    \ell(\bbeta, \btheta) &= \sum_{i=1}^n \log f(Y(\bs_i) \mid Y(\bs_{(i)}), \bbeta, \btheta) \nonumber, \\
    \log p(\bbeta, \btheta \mid \bY ) & = \ell(\bbeta, \btheta) + \log p(\bbeta, \btheta).
\end{align}
Hence the log-likelihood and log-posterior of the parameters $\left\{\bbeta, \btheta\right\}$ can be written as a sum of conditional normal log-densities, where the conditioning set is at most of size $m$. The cost of computing the log-posterior in \eqref{eq: vecchia_llk_repsonse} is linear in $n$ and cubic in $m$.

Although the Vecchia approximation reduces the complexity cost from $O(n^3)$ for the full likelihood to $O(nm^3)$, this can still pose challenges for very large $n$. We can further reduce the cost of Bayesian inference by using subsampling strategies which have had substantial success in SG methods \citep{newton2018stochastic}. Although we are still in correlated data setting, sampling the summands of \eqref{eq: vecchia_llk_repsonse} with equal probability and without replacement leads to an unbiased estimate of the gradient. Let $\mathcal{B} \subset \{1, \ldots, n\}$ be a subsample, \textit{i.e.}, a minibatch index set of size $n_{\mathcal{B}}$, and let
\begin{equation}\label{eq: mini_batch_llk}
    \bar{\ell}_{\mathcal{B}}(\bbeta, \btheta) = \frac{n}{\mbsize}\sum_{i\in \mathcal{B}} \log f(Y(\bs_i) \mid Y(\bs_{(i)}), \bbeta, \btheta).
\end{equation} 
\begin{theorem}\label{th: gradient_unbiased}
    The gradient of $\bar{\ell}_{\mathcal{B}}$ is an unbiased estimator of the gradient of the Vecchia posterior $\ell(\bbeta, \btheta)$.
\end{theorem}
\begin{proof}
\begin{align}\label{eq: unbiased_llk_gradients}
    \EX_{\mathcal{B}}[\nabla\bar{\ell}_{\mathcal{B}}(\bbeta, \btheta)] & = \nabla \EX_{\mathcal{B}} \left[ \frac{n}{\mbsize}\sum_{i=1}^n \log f(Y(\bs_i) \mid Y(\bs_{(i)}), \bbeta, \btheta)\delta_{i \in \mathcal{B}} \right] \nonumber \\
    & = \nabla  \sum_{i=1}^n \log f(Y(\bs_i) \mid Y(\bs_{(i)}), \bbeta, \btheta) \nonumber \\
    & =  \nabla\ell(\bbeta, \btheta).
\end{align}
\end{proof}
Using  \eqref{eq: unbiased_llk_gradients}, we can construct an unbiased estimate of the gradient of the Vecchia log-posterior based on a minibatch of the data:
\begin{equation}\label{eq: gradient_estimator}
    \bar{g}_{\mathcal{B}}(\bbeta, \btheta) =  \nabla\bar{\ell}_{\mathcal{B}}(\bbeta, \btheta) + \nabla\log p(\bbeta,\btheta),
\end{equation}hence reducing the cost of learning iterations to be linear in $\mbsize$ instead of $n$, \textit{i.e.}, $O(m^3\mbsize )$.

\section{The SG-MCMC Algorithm}\label{s:comp}

In this section we first review the general SG Langevin dynamics method and then present the proposed algorithm based on the Vecchia approximation.
\subsection{SG Langevin Dynamics}
SG Markov chain Monte Carlo \citep{Ma2015ACR} is a popular method for scalable Bayesian inference. SGMCMC proceeds by simulating continuous dynamics of a potential energy, namely the negative log-posterior $ -\log p(\bbeta, \btheta \mid \bY )$, such that the
dynamics generate samples from the posterior distribution. Let $\bphi = (\bbeta\Tr, \btheta\Tr)\Tr$ be the parameter vector concatenating the regression and covariance parameters of the Gaussian process regression model. The Langevin diffusion over $\log p(\bphi \mid \bY )$ is given by the stochastic differential equation (SDE)
\begin{equation}\label{eq: sgldsde}
    d(\bphi_t) = \nabla \log p(\bphi_t \mid \bY ) dt + \sqrt{2}dW_t,
\end{equation}
where $dW_t$ is Brownian motion and the index $t$ represents time. The distribution of samples $\bphi_t$ converges to the true posterior as $t \rightarrow \infty$ \citep{Roberts1998OptimalSO}. 

Since simulating a continuous time process is infeasible, in practice a discretized numerical approximation is used. Here we use the Euler discretization method. Let $h_t$ the step size at time $t$, and let $\bphi_t$ the current value of the parameter. The Euler approximation of the Langevin dynamics is given by
\begin{equation}\label{eq: euler_approx}
    \bphi_{t+1} = \bphi_t + h_t \nabla \log p(\bphi_t\mid \bY) + \sqrt{2h_t} e_t,
\end{equation}
where $e_t$ is random white noise. This recursive sampling approach is known as the Langevin Monte Carlo algorithm. Often, a Metropolis-Hastings (MH) correction step is added to account for the discretization error.

When the size of the dataset is large, computing the log-posterior gradient represents a computational bottleneck. To overcome this problem, the key idea of SGLD is to replace $ \nabla\log p(\bphi\mid \bY)$ with an unbiased gradient estimate, \textit{i.e.}, $ \bar{g}_{\mathcal{B}}(\bphi)$ in \eqref{eq: gradient_estimator} that is computationally cheaper to compute, and use a decreasing step size $h_t$ to avoid the costly M-H correction steps,
\begin{align}\label{eq: sgld_step}
    {\rm{SGLD:}}  \qquad  \bphi_{t+1} = \bphi_t + h_t  \bar{g}_{\mathcal{B}}(\bphi_t) +\sqrt{2h_t} e_t. & 
\end{align}
In order to assure convergence to the true posterior the  step sizes must satisfy
\begin{equation*}
    0 < h_{t+1} < h_t,\, \sum_{t=1}^{\infty}h_t = \infty, \mbox{ and } \sum_{t=1}^{\infty}h_t^2 < \infty.
\end{equation*}

The SGLD step in \eqref{eq: sgld_step} updates all parameters using the same step size. This can cause slow mixing when different parameters have different curvature or scales. SG Riemannian Langevin Dynamics (SGRLD) takes account of the difference in curvature and scale by using an appropriate Riemannian metric $G(\bphi)$, and simulates the diffusion by preconditioning the unbiased gradient and noise in \eqref{eq: sgld_step} using $G^{-1}(\bphi)$. SGRLD achieves better mixing by incorporating geometric information of the posterior. Commonly used metrics for $G(\bphi)$ include the Fisher information matrix and estimates of the Hessian of the log-posterior. Given a preconditioning matrix $G(\bphi)$, the SGRLD step is
\begin{equation}\label{eq: sgrld_step}
    {\rm{SGRLD}} \qquad \bphi_{t+1} = \bphi_t + h_t \left(G^{-1}(\bphi_t) \bar{g}_{\mathcal{B}}(\bphi_t) + \Gamma(\bphi_t) \right)\ + \sqrt{2h_t}G^{-1/2}(\bphi_t) e_t,
\end{equation}
where the term $\Gamma(\bphi_t)$ represents the drift term that describes how the preconditioner $G(\bphi_t)$ changes with respect to $\bphi_t$. The drift term is given by
\begin{equation}\label{eq :gamma_def}
    \Gamma(\bphi_t)_i = \sum_j \frac{\partial G(\bphi_t)_{ij}^{-1}}{\partial \bphi_{tj}}.
\end{equation}
The drift term vanishes in the SGLD step since the preconditioner is assumed to be the identity matrix.
The SGRLD algorithm in \eqref{eq: sgrld_step} takes steps in the steepest ascent on the manifold defined by the metric $G(\bphi_t)$. For many statistical models, the Fisher information matrix is intractable, however we will show in the next section that using the Vecchia's approximation we can compute the Fisher information and its inverse without incurring a high computational cost. Therefore, we use the Fisher information matrix, denoted $\mathcal{I}(\bphi)$, for $G(\bphi)$.

\subsection{Derivation of gradients and Fisher information for SGRLD}

Given an index set for a mini-batch subset of the data $\mathcal{B}$, the log-likelihood in \eqref{eq: mini_batch_llk} decomposes as the sum of log-conditional densities of the $Y(\bs_i)$ given the conditioning points $Y(\bs_{(i)})$. Computing the gradient of these conditional densities is analytically complicated and not computationally tractable. We follow \cite{guinness2019gaussian} to first rewrite the log-conditional densities in terms of marginal densities, and then compute the gradients and Fisher information. Let $u_i = Y(\bs_{(i)})$, the set of neighbours, and $v_i = (Y(\bs_{(i)}), Y(\bs_i))$, the vector of concatenating the $i^{{ \rm th}}$ observation and its neighbours. Let $\bQ_i \mbox{ and } \bR_i$ be the covariate matrices for $u_i$ and $v_i$ respectively, and let $\bA_i$ and $\bB_i$ denote the covariance matrices of $u_i$ and $v_i$. The minibatch log-likelihood in \eqref{eq: mini_batch_llk} can thus be written as
\begin{align}
    \bar{\ell}_{\mathcal{B}}(\bphi)  = & \sum_{i \in \mathcal{B}} \log f(v_i \mid \bphi) - \log f(u_i \mid \bphi) \nonumber \\
    = &-\frac{1}{2} \sum_{i \in \mathcal{B}} \log \det \bB_i - \log \det \bA_i \\
     & -\frac{1}{2} \sum_{i \in \mathcal{B}} [ (v_i - \bR_i\bbeta)^{\Tr}\bB_i^{-1}(v_i - \bR_i\bbeta) - 
    (u_i - \bQ_i\bbeta)^{\Tr}\bA_i^{-1}(u_i - \bQ_i\bbeta)] - \frac{\mbsize}{2}\log(2\pi). \nonumber
\end{align}
In order to compute the log-likelihood, we need the following quantities
\begin{align}
    p^{1}_{\mathcal{B}}(\btheta) & = \sum_{i \in \mathcal{B}} \log \det \bB_i - \log \det \bA_i \label{eq: llk_pieces-1}\\
    p^{2}_{\mathcal{B}}(\btheta) & = \sum_{i \in \mathcal{B}} (v_i\Tr\bB_i^{-1}v_i - u_i\Tr\bA_i^{-1}u_i) \label{eq: llk_pieces-2}\\
    p^{3}_{\mathcal{B}}(\btheta) & = \sum_{i \in \mathcal{B}} (\bR\Tr_i\bB_i^{-1}v_i - \bQ_i\Tr\bA_i^{-1}u_i) \label{eq: llk_pieces-3}\\
    p^{4}_{\mathcal{B}}(\btheta) & = \sum_{i \in \mathcal{B}} (\bR_i\Tr\bB_i^{-1}\bR_i - \bQ_i\Tr\bA_i^{-1}\bQ_i).  \label{eq: llk_pieces-4}
\end{align}
The quantities in \eqref{eq: llk_pieces-1} - \eqref{eq: llk_pieces-4} only depend on the covariance parameters $\btheta$ via $\bA_i \mbox{ and } \bB_i$ and not the mean parameters $\bbeta$. We can now write the minibatch log-likelihood as
\begin{equation}\label{eq: llk_form}
    \bar{\ell}_{\mathcal{B}}(\bphi) = -\frac{\mbsize}{2}\log(2\pi) - \frac{1}{2}\left[ p^{1}_{\mathcal{B}}(\btheta) + p^{2}_{\mathcal{B}}(\btheta) - 2\bbeta\Tr p^{3}_{\mathcal{B}}(\btheta) + \bbeta\Tr p^{4}_{\mathcal{B}}(\btheta)\bbeta \right].
\end{equation}

\subsubsection{Mean parameters}
The gradient of the minibatch log-likelihood with respect to the mean parameters $\bbeta$ is
\begin{equation}\label{eq: llk_beta_grad}
    \frac{\partial \bar{\ell}_{\mathcal{B}}(\bbeta, \btheta)}{\partial\bbeta} = p^{3}_{\mathcal{B}}(\btheta) - p^{4}_{\mathcal{B}}(\btheta)\bbeta.
\end{equation}
For the Fisher information, recall that if a random vector follows a multivariate normal model with mean and variance parameterized by two different parameter vectors, \textit{i.e.}, $W \sim \mathbf{N}(\mu(\bbeta), \Sigma(\btheta))$, then the Fisher information is block diagonal
$\mathcal{I}(\bphi) = \rm{diag}(\mathcal{I}(\bbeta), \mathcal{I}(\btheta)).$
Furthermore, let $J_{\bbeta}$ be the Jacobian of $\mu(\bbeta)$ with respect to $\bbeta$. Then the Fisher information matrix is analytically available \citep{MVN_fisher} and takes the form 
\begin{align}%\label{eq: fishe\bR_info_mvn}
    \mathcal{I}(\bbeta) &= J_{\bbeta}\Sigma^{-1}J_{\bbeta}\Tr \label{eq: fisher_info_mvn_beta} \\
    \mathcal{I}(\btheta)_{jk} &= \frac{1}{2} \Trace \left( \Sigma^{-1} \frac{\partial \Sigma}{\partial\btheta_j} \Sigma^{-1} \frac{\partial \Sigma}{\partial\btheta_k} \right). \label{eq: fisher_info_mvn_theta}
\end{align}

Using \eqref{eq: fisher_info_mvn_beta} and the chain rule property of the Fisher information, $\mathcal{I}_{Y(s_i)\mid u_i}(\bphi) = \mathcal{I}_{v_i}(\bphi)  - \mathcal{I}_{u_i}(\bphi) $ and summing over the components of the log-likelihood we get
\begin{equation}\label{eq: fisher_beta}
    \mathcal{I}_{\mathcal{B}}(\bbeta) = \sum_{i \in \mathcal{B}} (\bR_i\Tr\bB_i^{-1}\bR_i - \bQ_i\Tr\bA_i^{-1}\bQ_i) = p^{4}_{\mathcal{B}}(\btheta).
\end{equation}
Hence the Fisher information of $\bbeta$ is constant with respect to the mean parameters. In addition, since $\mathcal{I}(\bphi)$ is block diagonal, the drift term which represents how $\mathcal{I}(\bbeta)$ changes with respect to $\bphi$ is $\Gamma_{\mathcal{B}}(\bbeta) = {\boldsymbol{0}}_p$. The SGRLD step for regression parameters is thus
\begin{equation}\label{eq: gamma_beta}
    \bbeta_{t+1} = \bbeta_t + h_tp^{4}_{\mathcal{B}}(\btheta_t)^{-1}\left( p^{3}_{\mathcal{B}}(\btheta_t) - p^{4}_{\mathcal{B}}(\btheta_t)\bbeta_t\right) + \sqrt{2h_t}p^{4}_{\mathcal{B}}(\btheta)^{-1/2}e_t.
\end{equation}

\subsubsection{Covariance parameters}
For the covariance parameters, we first start by computing the partial derivatives of the quantities defined in \eqref{eq: llk_pieces-1}-\eqref{eq: llk_pieces-4} with respect to the components of $\btheta$, $p^{k}_j(\btheta) = \partial p^{k}_{\mathcal{B}}(\btheta)/\partial\theta_j \mbox{ for } j \in \left\{1, \ldots, 4\right\}$
\begin{align}
    p^{1}_j(\btheta) & = \sum_{i \in \mathcal{B}} \left(\Trace (\bB_i^{-1}\frac{\partial \bB_i}{\partial \theta_j}) - \Trace(\bA_i^{-1}\frac{\partial \bA_i}{\partial \theta_j}) \right)\label{eq: llk_pieces_d-1}\\
    p^{2}_j(\btheta) & = \sum_{i \in \mathcal{B}} \left(v_i\Tr \bB_i^{-1}\frac{\partial \bB_i}{\partial \theta_j}\bB_i^{-1}v_i - u_i\Tr \bA_i^{-1}\frac{\partial \bA_i}{\partial \theta_j}\bA_i^{-1}u_i\right) \label{eq: llk_pieces_d-2}\\
    p^{3}_j(\btheta) & = \sum_{i \in \mathcal{B}} \left(\bR\Tr_i\bB_i^{-1}\frac{\partial \bB_i}{\partial \theta_j}\bB_i^{-1}v_i - \bQ_i\Tr\bA_i^{-1}\frac{\partial \bA_i}{\partial \theta_j}\bA_i^{-1}u_i \right) \label{eq: llk_pieces_d-3}\\
    p^{4}_j(\btheta) & = \sum_{i \in \mathcal{B}} \left(\bR_i\Tr\bB_i^{-1}\frac{\partial \bB_i}{\partial \theta_j}\bB_i^{-1}\bR_i - \bQ_i\Tr\bA_i^{-1}\frac{\partial \bA_i}{\partial \theta_j}\bA_i^{-1}\bQ_i\right)  \label{eq: llk_pieces_d-4} \\
    \frac{\partial \bar{\ell}_{\mathcal{B}}(\bbeta, \btheta)}{\partial \theta_j} &= -\frac{1}{2} \left[ 
        p^{1}_j(\btheta) + p^{2}_j(\btheta) - 2p^{3}_j(\btheta)\bbeta + \bbeta\Tr p^{4}_j(\btheta) \bbeta
    \right]. \label{eq: llk_d_theta}
\end{align}
Using \eqref{eq: fisher_info_mvn_theta} and the chain rule decomposition of the Fisher information, we derive the analytic form of the Fisher information and drift term for the covariance parameters
\begingroup
\allowdisplaybreaks
    \begin{align}
    \mathcal{I}_{\mathcal{B}}(\btheta)_{jk} =& \frac{1}{2} \sum_{i \in \mathcal{B}} \Trace\left( \bB_i^{-1}\frac{\partial \bB_i}{\partial \theta_j}\bB_i^{-1}\frac{\partial \bB_i}{\partial \theta_k}\right) - \Trace\left(\bA_i^{-1}\frac{\partial \bA_i}{\partial \theta_j}\bA_i^{-1}\frac{\partial \bA_i}{\partial \theta_k} \right) \label{eq: fisher_theta} \\
    \frac{\partial\mathcal{I}_{\mathcal{B}}(\btheta)_{jk}}{\partial \theta_k}  = & \sum_{i \in \mathcal{B}} \Trace\left( \bB_i^{-1}\frac{\partial^2 \bB_i}{\partial \theta_j\partial \theta_k}\bB_i^{-1}\frac{\partial \bB_i}{\partial \theta_k}\right) - 
    \Trace\left( \bB_i^{-1} \frac{\partial \bB_i}{\partial \theta_j} \bB_i^{-1}\frac{\partial \bB_i}{\partial \theta_k}  \bB_i^{-1}\frac{\partial \bB_i}{\partial \theta_k} \right) \nonumber \\
    & - \sum_{i \in \mathcal{B}} \Trace\left( \bA_i^{-1}\frac{\partial^2 \bA_i}{\partial \theta_j\partial \theta_k}\bA_i^{-1}\frac{\partial \bA_i}{\partial \theta_k}\right) - 
    \Trace\left( \bA_i^{-1} \frac{\partial \bA_i}{\partial \theta_j} \bA_i^{-1}\frac{\partial \bA_i}{\partial \theta_k}  \bA_i^{-1}\frac{\partial \bA_i}{\partial \theta_k} \right) \label{eq: fisher_theta_d} \\
    \Gamma_{\mathcal{B}}(\btheta)_{j} = & -\sum_{k}  \mathcal{I}_{\mathcal{B}}(\btheta)_{j\cdot}^{-1} \frac{\partial\mathcal{I}_{\mathcal{B}}(\btheta)}{\partial \theta_k} \mathcal{I}_{\mathcal{B}}(\btheta)_{\cdot k}^{-1}.
\end{align}
\endgroup

\section{Simulation Study}\label{s:sim}
In this section, we test our proposed SGRLD method in \eqref{eq: sgrld_step} method on synthetic data and assess its performance against state-of-the-art Bayesian methods. We use Mean Squared Error (MSE) and coverage of credible intervals of posterior MCMC estimators to evaluate estimation of the spatial covariance parameters, and we use the Effective sample sizes (ESS) \citep{heidelberger1981spectral} per minute to gauge computational efficiency of MCMC algorithms. We present results only for the spatial covariance parameters $\btheta$ because the results are similar across methods for $\bbeta$.

\subsection{Data generation}\label{s:sim:gen}

We generate data on a regular rectangular grid formed with $n_1$ locations on the x-axis  and $n_2$ on the y-axis, with a total number of points $N = n_1 n_2$ and grid spacing one. We consider $N = \left\{ 10^4, 10^5, 10^6 \right\}$ for $n_1 = \{100, 300, 1000\}$ and $n_2 = N/n_1$. We generate the Gaussian process $Z(\bs)$ from a{\Matern} kernel with possible smoothness values 
$\nu \in \left\{ 0.5, 1.0, 1.5 \right\}$. The range parameter $\rho$ is chosen such that the correlation function is approximately $10^{-4}$ for the maximum distance between two points in the grid. We fix the spatial variance $\sigma^2 = 5$, and consider different scenarios for the observation 
noise based on the proportion of variance $\kappa = \tau^2/\sigma^2 \in \left\{ 0.2, 1.0, 5.0 \right\}$. Let $\bX_i = (1, x_i)$, the covariate for the $i^{{\rm th}}$ site, the mean of the Gaussian process will take the form  $\EX[Y(\bs_i)] = \beta_0 + \beta_1 \cos(x_{i})$, where $\beta_0 = -3, \mbox{ and } \beta_1 = 5$, and $x_i \iid \mbox{Uniform}(-3,3)$. 
For $N = 10^6$, generating a Gaussian process is computationally infeasible, thus we generate a Vecchia approximated Gaussian process with $m=120$ neighbors for each site. For each $N$, we generate 100 datasets and record the posterior mean and posterior credible intervals for each parameter.

\subsection{Competing methods and metrics}\label{s:sim:methods}
We compare our SGRLD method with four different MCMC methods. The first three are SG methods with adaptive drifts. The last method uses the full dataset to sample the posterior distribution using the Vecchia approximation. The three SGMCMC methods all use momentum and past gradient information to estimate the curvature and accelerate the convergence. These methods extend the momentum methods used in SG optimization methods for faster exploration of the posterior. The first method is Preconditioned SGLD (pSGLD) of \cite{li2016preconditioned} that uses the Root Mean Square Propagation (RMSPROP) \citep{hinton2012neural} algorithm to estimate a diagonal preconditioner for the minibatch gradient and injected noise. The second method is ADAMSGLD \citep{kim2022stochastic} that extends the widely used ADAM optimizer \citep{kingma2014adam} to the SGLD setting. ADAMSGLD approximates the first-order and second-order moments of the minibatch gradients to construct a preconditioner. Finally, we also include the performance of Momentum SGLD (MSGLD) where no preconditioner is used but past gradient information is used to accelerate the exploration of the posterior. The details of the above algorithms are included in the Appendix A.1. 
The final method we consider is the Nearest Neighbor Gaussian Process (NNGP) method \citep{datta2016a}. This method is the standard MCMC method based on the Vecchia approximation and is implemented in the {\rm{ R}} package {\rm{spNNGP}} \citep{spNNGP}. For this method, the initial values are set to the true values and the Metropolis-Hastings proposal distribution is chosen adaptively using the default settings.

For the SGMCMC methods, the batch size is set to $250$ when the number of location is $10^4$ and $500$ for the other two cases. We noticed during our experiments that batch sizes in the order of $200$ perform better than smaller size ones, with very similar performance to larger ones. The number of epochs will depend on the size of data, and is chosen such that the total number of iterations is $20000$, of which a quarter are discarded as burn-in.  The learning rate is divided by a factor of $2$ every $5$ epochs, so the final learning rate is set at $1\%$ of the initial value. A first tentative value of the learning rate is set at $1/N$, then reduced until the norm of the first step is less than one. We noticed that the appropriate learning rate for our SGRLD method is within one to two orders of magnitude large than the learning rate for the other SG sampling methods. For all the methods, the size of the conditioning set is fixed at $m=15$. The conditioning sets were selected using the max-min ordering \citep{katzfuss2021general} for $N<10^6$, and random ordering otherwise. \cite{katzfuss2021general} showed that the max-min ordering results in significant improvements over other coordinate based orderings. However, when $N$ is very large, the cost of max-min ordering becomes prohibitive. For the NNGP method, we take $2000$ samples when $N <10^5$ and $1000$ otherwise. For all the methods we use a non-informative flat prior on the regression parameters. For the covariance parameters, we set the following priors:
\begin{align*}
    \rho &\sim \mbox{Gamma}(9.0, 2.0) \\
    \nu & \sim \mbox{Log-Normal}(1.0, 1.0) \\
   \tau^2,\, \sigma^2 &\sim \mbox{Gamma}(0.1, 0.1) 
\end{align*}
The prior $90\%$ credible intervals for $\rho \mbox{ and } \nu$ are $(2.06,7.88)$ and $(0.52, 14.08)$ respectively, which represent weakly informative priors.

\subsection{Results}\label{s:sim:results}
Table \ref{tab: sim_all_mse} gives the MSE results. Our SGRLD method outperforms all the others with very low MSE across parameters. In particular, the SGMCMC methods all outperform the NNGP method. In our experiments, we noticed that the NNGP method suffers from very slow mixing due to  the M-H step necessary for sampling the covariance parameters. In fact, even if we start the NNGP sampling process at true values of the covariance parameters, and reduce the variance of the proposal distribution, the acceptance rate of the M-H step stays below $15\%$. None of the SGMCMC methods requires any such step as long as the learning rate is kept small. 
%In Table-\ref{tab: sim_n_1e5} and Table-\ref{tab: sim_n_1e5_ess} the results for $N=10^5$ are very similar in terms of MSE and ESS to the previous case $N=10^4$. Here again, SGRLD clearly outperforms all the other methods in both estimation accuracy and effective samples per unit time. We also observe a slightly more ESS discrepancy between the variance and the other parameters for the other two SG based methods. Furthermore, the low number of NNGP effective samples remains, while the MSE improved but still considerably higher than the other methods.
\begin{table}
    \centering
    \renewcommand{\arraystretch}{0.8}%
    \caption{Mean squared error (Monte Carlo standard errors) of covariance parameters computed using 100 simulations, each having sample size $N$. The proposed SGRLD method compared with other SGMCMC methods (pSGLD, ADAMSGLD, MSGLD) and the full likelihood NNGP method.}\label{tab: sim_all_mse}%
    %\resizebox{\textwidth}{!}{
    \begin{tabular}{ c|c|c|c|c|c }
        %\hline
        N & Algorithm& Variance ($\sigma^2$) & Range ($\rho$) & Smoothness ($\nu$) & Nugget ($\tau^2$)  \\ 
        \hline
       \multirow{5}{*}{$10^4$} & pSGLD & $0.074 (0.013) $ & $0.039 (0.008)$  & $0.103(0.017)$ & $0.002(4\cdot10^{-4})$ \\ 
       & ADAMSGLD & $0.075 (0.017) $ & $0.036 (0.008)$  & $0.129(0.023)$ & $0.002(6\cdot10^{-4})$ \\
       & MSGLD & $0.066 (0.014) $ & $0.034 (0.008)$  & $0.108(0.0196)$ & $0.002(6\cdot10^{-4})$ \\
       & NNGP & $0.414 (0.131) $ & $0.095 (0.071)$  & $0.162(0.106)$ & $0.093(2.4\cdot10^{-2})$ \\
       & SGRLD & $0.056 (0.016) $ & $0.031 (0.006)$  & $0.077(0.013)$ & $0.001(10^{-4})$ \\ \hline
       \multirow{5}{*}{$10^5$} & pSGLD & $0.008 (0.001) $ & $0.002 (0.0003)$  & $0.011(0.0019)$ & $1\cdot10^{-4}(2\cdot10^{-5})$ \\ 
       & ADAMSGLD & $0.014 (0.005) $ & $0.008 (0.002)$  & $0.031(0.008)$ & $1\cdot10^{-4}(2\cdot10^{-4})$ \\
       & MSGLD & $0.017 (0.001) $ & $0.003 (5\cdot10^{-4)}$  & $0.019(0.002)$ & $2\cdot10^{-4}(4\cdot10^{-5})$ \\
       & NNGP & $0.116 (0.030) $ & $0.024 (0.01)$  & $0.118(0.08)$ & $4\cdot 10^{-2}(0.01)$ \\
       & SGRLD & $0.005 (8\cdot 10^{-4}) $ & $0.001 (1.0\cdot 10^{-4})$  & $0.008(1.8 \cdot 10^{-3})$ & $10^{-4}(2\cdot10^-5)$ \\\hline
       \multirow{5}{*}{$10^6$} & pSGLD & $0.003 (0.001) $ & $0.003 (0.0008)$  & $0.002(0.0014)$ & $3.1\cdot10^{-4}(6\cdot10^{-5})$ \\ 
       & ADAMSGLD & $0.009 (0.002) $ & $0.006 (0.002)$  & $0.026(0.007)$ & $2\cdot10^{-4}(9\cdot10^{-5})$ \\
       & MSGLD & $0.011 (1.8\cdot10^{-3}) $ & $0.003 (5\cdot10^{-4)}$  & $0.019(0.002)$ & $1\cdot10^{-5}(3\cdot10^{-5})$ \\
       & NNGP & $0.078 (0.055) $ & $0.016 (0.009)$  & $0.126(0.086)$ & $0.08(0.049)$ \\
       & SGRLD & $0.002 (3\cdot 10^{-4}) $ & $0.001 (1\cdot 10^{-4})$  & $0.004(6.1 \cdot 10^{-3})$ & $0.4\cdot 10^{-4}(1\cdot10^-5)$ \\
    \end{tabular}%}
\end{table}

Table \ref{tab: sim_all_coverage} summarizes the results for the coverage of the $95\%$ credible intervals. Our SGRLD method again outperforms the other methods. One exception is that the pSGLD algorithm surpasses the SGRLD in the coverage of the variance parameter. Across methods, the smoothness parameter consistently has the lowest coverage, followed by the range parameter. Even for $N=10^6$, MSGLD, ADAMSGLD and NNGP fail to attain attain a $90\%$ coverage rate. Whilst the SGRLD coverage rate for both parameters is higher than $90\%$ even for $N=10^4$.

\begin{table}
\caption{Coverage of the $95\%$ credible intervals (Monte Carlo standard errors) for the covariance parameters computed using 100 simulations, each having sample size $N$. The proposed SGRLD method is compared with other SGMCMC methods (pSGLD, ADAMSGLD, MSGLD) and the full likelihood NNGP method.}\label{tab: sim_all_coverage}
\centering
\renewcommand{\arraystretch}{0.8}%
\begin{tabular}{c|c|c|c|c|c}
N &Algorithm & Variance, $\sigma^2$ & Range, $\rho$ & Smoothness, $\nu$ & Nugget, $\tau^2$ \\
\hline
\multirow{5}{*}{$10^4$} & pSGLD & $0.977(0.02)$ & $0.845(0.06)$ & $0.815(0.06)$ & $0.931(0.05)$\\
& ADAMSGLD & $0.886 (0.05)$ & $0.791(0.08)$ & $0.647(0.08)$ & $0.636(0.05)$\\
& MSGLD & $0.793(0.03)$ & $0.847 (0.07)$ & $0.709(0.07)$ & $0.683(0.05)$\\
& NNGP & $0.783 (0.06)$ & $0.776(0.05)$ & $0.614(0.07)$ & $0.812(0.01)$ \\
& SGRLD & $0.955(0.03)$ & $0.924(0.05)$ & $0.909(0.04)$ & $0.935(0.01)$\\ \hline
\multirow{5}{*}{$10^5$} & pSGLD & $0.991(0.03)$ & $0.913(0.04)$ & $0.862(0.05)$ &$0.965(0.02)$\\
& ADAMSGLD & $0.861(0.03)$ & $0.754(0.07)$ & $0.814(0.03)$ & $0.738(0.05)$\\
& MSGLD & $0.896(0.04)$ & $0.881(0.07)$ & $0.774(0.08)$ & $0.872(0.07)$\\
& NNGP & $0.826(0.05)$ & $0.758(0.04)$ & $0.714(0.03)$ & $0.872(0.02)$\\
& SGRLD & $0.957(0.01)$ & $0.964(0.01)$ & $0.948(0.01)$ & $0.932(5\cdot10^{-3})$\\ \hline
\multirow{5}{*}{$10^6$} & pSGLD & $0.987(6\cdot10^{-3})$ & $0.934(0.02)$ & $0.901(0.03)$ & $0.961(0.01)$\\
& ADAMSGLD & $0.902(0.01)$ & $0.824(10^{-3})$ & $0.838(0.02)$ & $0.781(0.03)$\\
& MSGLD & $0.884(10^{-3})$ & $0.918(0.02)$ & $0.846(0.01)$ & $0.926(0.01)$\\
& NNGP & $0.866(0.03)$ & $0.818(0.06)$ & $0.834(0.04)$ & $0.862(0.01)$\\
& SGRLD & $0.968(6\cdot10^{-3})$ & $0.941(8\cdot10^{-3})$ & $0.929(5\cdot10^{-3})$ & $0.941(2\cdot10^{-3})$\\
\end{tabular}
\end{table}

For the ESS results in Table \ref{tab: sim_all_ess},  the SGRLD method offers superior effective samples per unit time for all the parameters. The pSGLD and MSGLD method seem to adapt to the curvature of the variance parameter, with pSGLD offering higher effective samples than SGRLD. This suggests that the computed preconditioner in pSGLD adapts mainly to the curvature of the variance term, but fails to measure the curvature of the smoothness and range. A similar behavior is also observed in the other two methods, MSGLD and ADAMSGLD. On the other hand, the ESS for SGRLD is of the same order for all the parameters. We believe this indicates that using the Fisher information matrix as a Riemannian metric provides an accurate measure of the curvature and results in higher effective samples for all the parameters. The NNGP method provides low effective sample sizes compared to the other three methods due to the low acceptance rate from the MH correction step.

\begin{table}
    \centering
    \renewcommand{\arraystretch}{0.8}%
    \caption{Effective sample size per minute (Monte Carlo standard errors) of covariance parameters computed using 100 simulations, each having sample size $N$. The proposed SGRLD method is compared with other SGMCMC methods (pSGLD, ADAMSGLD, MSGLD) and the full likelihood NNGP method.
    \label{tab: sim_all_ess}}%
    %\resizebox{\textwidth}{!}{
    \begin{tabular}{ c|c|c|c|c|c }
        N &Algorithm & Variance, $\sigma^2$ & Range, $\rho$ & Smoothness, $\nu$ & Nugget, $\tau^2$   \\ 
        \hline
        \multirow{5}{*}{$10^4$} & pSGLD & $42.97 (1.57) $ & $8.43 (0.54)$  & $4.33(0.26)$ & $9.82(0.79)$ \\ 
        & ADAMSGLD & $9.12 (0.45) $ & $4.22 (0.33)$  & $2.85(0.28)$ & $3.80(0.48)$ \\
        & MSGLD & $15.68 (0.95) $ & $6.48 (0.70)$  & $3.65(0.44)$ & $5.11(0.78)$ \\
        & NNGP & $1.02 (0.33) $ & $0.99 (0.24)$  & $1.11(0.75)$ & $0.51(0.14)$ \\
        & SGRLD & $23.8(1.15) $ & $23.9 (1.19)$  & $25.2(1.25)$ & $30.5(1.55)$ \\ \hline
        \multirow{5}{*}{$10^5$} & pSGLD & $66.87 (2.09) $ & $10.06 (0.65)$  & $3.59(0.21)$ & $11.3(0.79)$ \\ 
        & ADAMSGLD & $7.87 (0.38) $ & $2.37 (0.27)$  & $1.15(0.13)$ & $1.64(0.24)$ \\
        & MSGLD & $12.92 (0.67) $ & $3.15 (0.36)$  & $1.206(0.11)$ & $1.71(0.13)$ \\
        & NNGP & $0.89 (0.08) $ & $0.75 (0.31)$  & $1.02(0.14)$ & $0.47(0.07)$ \\
        & SGRLD & $22.7(0.33) $ & $22.44 (0.27)$  & $22.69(0.13)$ & $23.23(0.34)$ \\ \hline
        \multirow{5}{*}{$10^6$} & pSGLD & $96.49 (3.37) $ & $13.68 (0.81)$  & $3.04(0.11)$ & $9.74(0.42)$ \\ 
        & ADAMSGLD & $6.17 (0.13) $ & $4.56 (0.52)$  & $1.98(0.62)$ & $2.36(0.83)$ \\
        & MSGLD & $15.07 (1.01) $ & $3.78 (0.81)$  & $2.06(0.30)$ & $5.01(0.97)$ \\
        & NNGP & $0.81 (0.16) $ & $1.01 (0.34)$  & $0.28(0.05)$ & $0.52(0.03)$ \\
        & SGRLD & $25.8(0.14) $ & $26.05 (0.18)$  & $29.62(0.28)$ & $24.07(0.27)$ \\
    \end{tabular}%}
\end{table}

Given the performance of the SG based methods in this simulation study, especially the SGRLD, we conducted an additional simulation study where we focus on point estimates instead of fully Bayesian inference. In Appendix A.2, we tweak the SGRLD method and turn it into a SG Fisher scoring (SGFS) algorithm for point estimates. We compare this method to the full data gradient Fisher scoring method \citep{guinness2019gaussian} already implemented in the {\rm{GpGp} R} package \citep{guinness2018gpgp}. We find improved speed and estimation precision compared to the GpGp package.

\section{Analysis of Global Ocean Temperature Data}\label{s:argo}

We apply the proposed method to the ocean temperature data provided by the Argo Program \citep{argo} made available through the {\tt GpGp} package \citep{guinness2018gpgp}.  Each of the $n=32,436$ observations are taken on buoys in the Spring of 2016.  Each observation measures of ocean temperature (C) at depths of roughly 100, 150 and 200 meters. The data are plotted in Figure \ref{f:argo} for depth 100 meters; we analyze these data using the methods evaluated in Section \ref{s:sim}. As an illustrative example, the mean function is taken to be quadratic in latitude and longitude and the covariance function is the isotropic $\Matern$ covariance function used in Section \ref{s:sim}. All prior distributions and MCMC settings are the same as in Section \ref{s:sim}.

\begin{figure}
    \centering
    \includegraphics[scale=1,trim={0 5cm 0 5cm},clip,page=2]{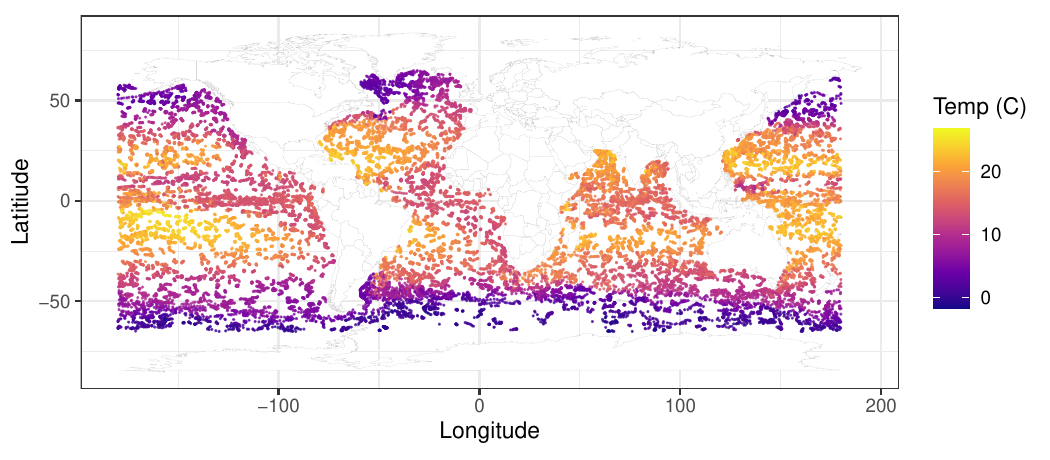}
    \caption{Argo ocean temperature measurements at a depth of 100 meters.}
    \label{f:argo}
\end{figure}

We first split the data into a test and training set, keeping $20\%$ of the observations in the testing set.  We train the models using $8000$ and $40000$ MCMC iterations for the NNGP ad SGRLD method respectively. For the SGRLD method this requires only $400$ epochs. We compare our SGRLD with the NNGP method using prediction MSE, squared correlation between predicted and observed ($R^2$) and coverage of 95\% prediction intervals on the test set. We also include the effective sample size per minute for all the model parameters.

\begin{table}
    \centering
    \begin{tabular}{c|cccc}
    \multicolumn{1}{c}{} & MSE & Coverage & $R^2$ & Time (in minutes) \\
    \hline
    NNGP  & $6.41$ & $0.88$ & $0.89$ &  $218.55$\\
    SGRLD & $1.47$ & $0.93$ & $0.94$ &  $7.01$   
    \end{tabular}
    \caption{Prediction Mean Squared Error (MSE), squared correlation between predicted and observed ($R^2$) and coverage rate of the $95\%$ predictive credible intervals on the test set and the correlation between the predicted temperatures and true observed values. The last column gives the total training time in minutes. We take $8000$ and $40000$ samples using the NNGP and SGRLD method respectively. }
    \label{tab: argo2016_test_results}
\end{table}

Table \ref{tab: argo2016_test_results} gives the MSE and coverage rate  on the testing set, and total training time respectively. Our method achieves less than the quarter of the MSE of NNGP while also requiring less than a twentieth of the time. For the coverage of the  $95\%$ prediction intervals, the NNGP method's average coverage on the testing set is significantly lower than the nominal value, while our proposed method achieves $93\%$ coverage.

\begin{table}
    \centering
    \renewcommand{\arraystretch}{0.8}
    \begin{tabular}{c|cccc}
    \multicolumn{1}{c}{Method}& Parameter & Posterior mean & $95\%$ CI & ESS/min \\
    \hline
    \multirow{4}{*}{NNGP} &$\sigma^2$& $6.72$ & $(6.32, 7.08)$ & $0.17$ \\
                          &  $\rho$  & $0.10$ & $(0.10, 0.11)$ & $3.42$ \\
                          &   $\nu$  & $0.33$ & $(0.32, 0.34)$ & $0.04$ \\
                          & $\tau^2$ & $0.08$ & $(0.08, 0.09)$ & $0.08$ \\
    \hline
    \multirow{4}{*}{SGRLD} &$\sigma^2$& $10.64$ & $(7.41, 13.57)$ & $52.21$ \\
                          &  $\rho$  & $48.93$  & $(22.94, 68.46) $ & $115.41$ \\
                          &   $\nu$  & $0.25$   & $(0.23, 0.27)$ & $18.68$ \\
                          & $\tau^2$ & $0.04$   & $(0.03, 0.05)$ & $39.13$ \\
    \end{tabular}
    \caption{Posterior mean, $95\%$ credible intervals and effective sample size per minute for all the covariance parameters.}
    \label{tab: argo2016_theta_posterior}
\end{table}

\begin{figure}
    \centering
    \includegraphics[width = .9\textwidth, ]{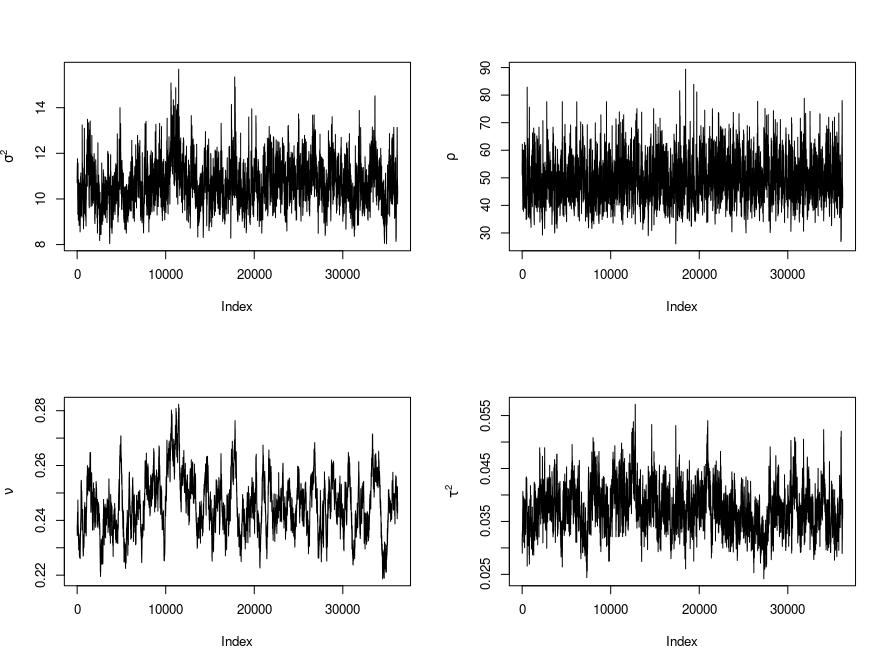}
    \caption{Evolution of SGRLD sampling from the posterior distribution of the covariance parameters.}
    \label{fig:sgrld_mcmc}
\end{figure}

Table \ref{tab: argo2016_theta_posterior} gives the posterior mean, 95\% interval and effective sample size per minute for the covariance parameters for SGRLD and NNGP. The posterior means and credible intervals for the range parameter $\rho$, and to a lesser extent the spatial variance $\sigma^2$, vary substantially between the two methods. The range estimates from SGRLD are almost three orders of magnitude higher than the NNGP estimate. Given the prediction results in Table \ref{tab: argo2016_test_results}, this indicates that the NNGP method is underestimating the range parameter. Furthermore, for NNGP, the credible interval for the range parameter has a total width of $10^{-2}$, perhaps indicating poor convergence. %; this added to the low MH acceptance rate suggests lack of posterior exploration or a local mode collapse.  
We also see from Table \ref{tab: argo2016_theta_posterior} that our SGRLD method allows fast exploration of the posterior and leads to almost $500$ times more effective samples per unit time, while giving reasonable convergence (Figure \ref{fig:sgrld_mcmc}). 

\begin{table}
    \centering
    \renewcommand{\arraystretch}{1.2}
    \begin{tabular}{cc|cccc}
    $\mbsize$ & $m$        & $\sigma^2$ & $\rho$ & $\nu$ & $\tau^2$ \\
    \hline
    \multirow{3}{*}{$100$} & $10$ & $10.18_{(8.79, 11.89)} $& $53.67_{(41.11, 66.87)}$ & 
                           $0.24_{(0.22, 0.25)}$ & $0.04_{(0.03, 0.04)}$  \\
                           & $15$ & $11.29_{(9.08, 13.53)}$ & $54.95_{(39.18, 72.83)}$ & 
                           $ 0.25_{(0.22, 0.27)} $ & $0.04_{(0.03, 0.04)}$ \\ 
                           & $30$ & $ 9.60_{(5.09, 13.24)}$ & $46.41_{(13.34, 74.52)}$ & $0.24_{(0.20, 0.26)}$ 
                           & $0.04_{(0.03, 0.07)}$  \\ \hline
    \multirow{3}{*}{$250$} & $10$ & $10.52_{(9.08, 12.25)}$ & $49.55_{(37.79, 62.85)}$ &
                           $0.25_{(0.22, 0.26)}$ & $0.04_{(0.03, 0.05)}$ \\ 
                           & $15$ & $10.64_{(7.41, 13.57)}$ & $48.93_{(22.94, 68.46)}$ & $0.25_{(0.23, 0.27)}$ 
                           & $0.04_{(0.03, 0.05)}$  \\
                           & $30$ & $10.59_{(7.41, 13.57)}$ & $46.08_{(22.95, 69.46)}$ & 
                           $0.25_{(0.23, 0.27)}$ & $0.04_{(0.03, 0.05)}$  \\ \hline
    \multirow{3}{*}{$500$} & $10$ & $11.02_{(9.74, 12.69)}$ & $49.23_{(39.56, 60.36)}$ & 
                           $0.25_{(0.23,0.27)}$ & $0.04_{(0.03, 0.04)}$  \\
                           & $15$ & $11.41_{(9.75, 13.20)}$ & $48.42_{(37.98, 60.26)}$ & $0.25_{(0.24, 0.27)}$ & $0.04_{(0.03, 0.04)}$  \\ 
                           & $30$ & $11.52_{(9.72, 13.61)}$ & $48.39_{(36.61, 62.35)}$ & $0.26_{(0.24, 0.28)}$ &$0.04_{(0.03, 0.04)}$
    \end{tabular}
    \caption{Sensitivity analysis to the choice of the conditioning set size $m$ and the mini-batch size $\mbsize$. Posterior mean and $95\%$ credible intervals are displayed for each combination of $\mbsize \mbox{ and } m$.}
    \label{tab: batch_m_sensitivity}
\end{table}

Finally, as a sensitivity analysis, we compare the SGRLD results with mini-batch size $\mbsize \in \{100, 250, 500\}$ and conditioning set size $m \in \{10,15,30\}$.  Table \ref{tab: batch_m_sensitivity} show the posterior mean and $95\%$ credible intervals of the covariance parameters for all combinations of the two hyperparameters. The posterior mean of the spatial variance, smoothness and nugget vary little across these combinations of tuning parameters. For the range parameter, we notice a sensitivity to small batch sizes, \textit{e.g.}, $\mbsize=100$ resulting in wide credible intervals and larger estimates compared to the other cases. For batch sizes $\{250, 500\}$ the estimates are similar across values of $m$.   

\section{Discussion}\label{s:discussion}
SG methods offer considerable speed-ups when the data size is very large. In fact, one can take hundreds or even thousands of steps in one pass through the whole dataset in the time it takes for only one step if the full dataset is used. This enables fast exploration of the posterior in significantly less time. GPs however fall within the correlated setting case where SGMCMC methods have received limited attention. Spatial correlation is a critical component of GPs and naive subsampling during parameter estimation would lead to random divisions of the spatial domain at each iteration. By leveraging the form of the Vecchia approximation, we derive unbiased gradient estimates based on minibatches of the data. We developed a new stochastic gradient based MCMC algorithm for scalable Bayesian inference in large spatial data settings.  Without the Vecchia approximation, subsampling strategies would always lead to biased gradient estimates. The proposed method also uses the exact Fisher information to speed up convergence and explore the parameter space efficiently. Our work contributes to the literature on scalable methods for Gaussian process, and can be extended to non Gaussian models \textit{i.e.} classification.

\section*{Acknowledgements}
This research was partially supported by National Science Foundation grants DMS2152887 and DMR-2022254, and by grants from the Southeast National Synthesis Wildfire and the United States Geological Survey’s National Climate Adaptation Science Center (G21AC10045).
\clearpage
\begin{singlespace}
	\bibliographystyle{rss}
	\bibliography{refs}
\end{singlespace}

\section*{Appendix A.1: Computational Details}\label{s:A1}
Here we give the detailed algorithms of the SG methods with adaptive drifts.
The RMSprop (Root Mean Square Propagation) algorithm is an optimization algorithm originally developped for training neural networks models. It adapts the learning rates of each parameter based on the historical gradient information. This can be seen as adaptive preconditioning method.

\begin{algorithm}[H]
\SetAlgoNlRelativeSize{0}
\SetAlgoNlRelativeSize{-1}
\KwIn{Initial parameter values $\theta_0$, learning rate $h_0$, decay rate $\rho$, small constant $\epsilon$}
\KwOut{Optimized parameter values $\theta$}
Initialize square gradient accumulator $r_0 = 0$\;
\While{not converged}{
    Sample minibatch without repetition;
    Compute gradient $\bar{g}$ on mini-batch\;
    Accumulate squared gradient: $r_t \leftarrow \rho r_{t-1} + (1 - \rho) \bar{g}\odot \bar{g} $\;
    Update parameters: $\theta_{t+1} \leftarrow \theta_{t} - h_t \bar{g} \oslash \sqrt{r_t + \epsilon} $\;
}
\caption{RMSprop Algorithm \label{alg: rmsprop}}
\end{algorithm}

\vspace{20pt}
Momentum SGD is an optimization algorithm that uses a Neseterov momentum term to accelerate the convergence in the presence of high curvature or noisy gradients. Momentum SGD proceeds as follows

\begin{algorithm}[H]
\SetAlgoNlRelativeSize{0}
\SetAlgoNlRelativeSize{-1}
\KwIn{Initial parameter values $\theta_0$, learning rate $h_0$, momentum term $\alpha$}
\KwOut{Optimized parameter values $\theta$}
Initialize velocity $v_0 = 0$\;
\While{not converged}{
Sample minibatch without repetition;
    Compute gradient $\bar{g}_t$ on mini-batch\;
    Update velocity: $v_t \leftarrow \alpha v_{t-1} - h_t \bar{g}$\;
    Update parameters: $\theta_{t+1} \leftarrow \theta_t + v_t$\;
}
\caption{Momentum SGD Algorithm\label{alg: msgd}}
\end{algorithm}

\vspace{25pt}
The Adam algorithm combines ideas from RMSprop and momentum to adaptively adjust learning rates.

\begin{algorithm}[H]
\SetAlgoNlRelativeSize{0}
\SetAlgoNlRelativeSize{-1}
\KwIn{Initial parameter values $\theta_0$, learning rate $h_0$, exponential decay rates for moments $\alpha_1$, $\alpha_2$, small constant $\epsilon$}
\KwOut{Optimized parameter values $\theta$}
Initialize moment estimates $m_0 = 0$, $v_0 = 0$, time step $t = 0$\;
\While{not converged}{
    Sample minibatch without repetition;
    Compute gradient $\bar{g}$ on mini-batch\;
    Update biased first moment estimate: $m_{t+1} \leftarrow \alpha_1 m_t + (1 - \alpha_1) \bar{g}$\;
    Update biased second raw moment estimate: $v_{t+1} \leftarrow \beta_2 v + (1 - \alpha_2) \bar{g}
    \odot \bar{g}$\;
    Correct bias in moment estimates: $\hat{m}_{t} \leftarrow m_{t}/(1 - \alpha_1^t)$, $\hat{v}_{t} \leftarrow v_t/(1 - \alpha_2^t)$\;
    Update parameters: $\theta_{t+1} \leftarrow \theta_t - \alpha \hat{m}_t\oslash (\sqrt{\hat{v}_t} + \epsilon)$\;
}
\caption{Adam Algorithm\label{alg: adam}}
\end{algorithm}

\section*{Appendix A.2: Additional Results}\label{s:A2}

\subsection*{Maximum likelihood estimates}\label{A:sim:results}

As discussed in the simulation study section, a small modification to the SGRLD algorithm yields a stochastic Fisher scoring method for the likelihood using the Vecchia approximation. To perform Fisher scoring, we only need to remove the injected noise and drift terms from the updates in equations \eqref{eq: sgrld_step}. Let $\bbeta_t, \btheta_t \mbox{ and } h_t$, the current values of the parameters and the step size respectively. The SGFS scoring updates are
\begin{align}\label{eq: sgfs}
    \bbeta_{t+1} &= \bbeta_t + h_tp^{4}_{\mathcal{B}}(\btheta_t)^{-1}\left( p^{3}_{\mathcal{B}}(\btheta_t) - p^{4}_{\mathcal{B}}(\btheta_t)\bbeta_t\right) \\
    \btheta_{t+1} &= \btheta_t + h_t \mathcal{I}_{\mathcal{B}}(\btheta)^{-1}\nabla_{\btheta_t}\bar{\ell}_{\mathcal{B}}(\bbeta_t, \btheta_t),
\end{align}
where $\mathcal{I}_{\mathcal{B}}(\btheta)$ is given in \eqref{eq: fisher_theta} and $\nabla_{\btheta_t}\bar{\ell}_{\mathcal{B}}(\bbeta_t, \btheta_t)$ is the vector with elements defined in \eqref{eq: llk_d_theta}.

We compare the SGFS method to the full data Fisher scoring method in \cite{guinness2019gaussian} and the most widely used SGD variants. We limit the setting to $N=10^4$ locations, the learning rate scheduling and batch size dimensions are kept the same, and the number of epochs is set to $10$. To avoid overfitting in the SGD methods we use the stopping rule proposed by \cite{chee2018convergence}. This method keeps a running average of the inner product of successive gradients and detects when this quantity changes sign. The results are summarized in Table-\ref{tab: sim_n_1e4gpgp}.
\begin{table*}[!htbp]
    \centering
    \caption{Mean squared error of covariance parameters from $10^4$ locations. Mean and standard deviation (in parenthesis) displayed over 100 simulations.
    \label{tab: sim_n_1e4gpgp}}%
    %\resizebox{\textwidth}{!}{
    \begin{tabular}{ c|c|c|c|c|c }
        \hline
        Algorithm& Variance & Range & Smoothness  & Nugget & Time (in seconds) \\ 
        \hline
        RMSPROP & $0.12 (0.03) $ & $0.39 (0.069)$  & $0.72(0.12)$ & $0.02(3\cdot10^{-4})$ & $9.26$ \\ 
        ADAM & $0.52 (0.07) $ & $0.69 (0.02)$  & $0.93(0.06)$ & $0.11(3\cdot10^{-3})$ & $11.71$\\
        MSGD & $0.19 (0.06) $ & $0.48 (0.10)$  & $1.08(0.09)$ & $0.17(8\cdot10^{-3})$ & $10.05$\\
        GpGp & $0.92 (0.16) $ & $0.09 (0.02)$  & $1.04(0.25)$ & $0.21(6\cdot10^{-3})$ & $38.72$\\
        SGFS & $0.06 (0.014) $ & $0.10 (0.04)$  & $0.16(0.01)$ & $0.11(2\cdot10^{-2})$ & $14.03$\\
    \end{tabular}
\end{table*}

The results in Table-\ref{tab: sim_n_1e4gpgp} show that the SGFS outperforms the other methods in terms of estimation error. Compared to GpGp, the stochastic methods take at most half the time while performing twenty times more iterations.
% \begin{table*}[htbp]
%     \centering
%     \caption{Mean squared error of covariance parameters from $10^5$ locations. Mean and standard deviation displayed over 100 simulations}
%     \label{tab: sim_n_1e5gpgp}%
%     \resizebox{\textwidth}{!}{
%     \begin{tabular}{ |c|c|c|c|c| }
%         \hline
%         {Algorithm}& Variance & Range & Smoothness  & Nugget  \\ [2pt]
%         \hline
%         {RMSPROP} & $0.010 (0.001) $ & $0.24 (0.0003)$  & $0.95(0.0019)$ & $3\cdot10^{-4}(2\cdot10^{-5})$ \\ [2pt]
%         \hline
%         {ADAM} & $0.014 (0.005) $ & $0.007 (0.002)$  & $0.021(0.008)$ & $1\cdot10^{-4}(2\cdot10^{-4})$ \\[2pt]
%         \hline
%         {MSGD} & $0.01 (1.8\cdot10^{-3}) $ & $0.003 (5\cdot10^{-4)}$  & $0.019(0.002)$ & $2\cdot10^{-4}(4\cdot10^{-5})$ \\[2pt]
%         \hline
%         {GPGP} & $0.96 (0.33) $ & $0.25 (0.11)$  & $0.99(0.15)$ & $1.01(0.24)$ \\[2pt]
%         \hline
%         {SGFS} & $0.009 (8\cdot 10^{-4}) $ & $0.003 (0\cdot 10^{-4})$  & $0.01(1.8 \cdot 10^{-3})$ & $10^{-4}(2\cdot10^-5)$ \\[2pt]
%         \hline
%         
%     \end{tabular}}
% \end{table*}

\end{document}